\documentclass[fleqn,10pt]{wlscirep}
\usepackage[utf8]{inputenc}
\usepackage[T1]{fontenc}
\title{Inferring Causality from Time Series data based on Structural Causal Model and its application to Neural Connectomics}

\author[1,*]{Rahul Biswas}
\author[2]{SuryaNarayana Sripada}
\author[3]{Somabha Mukherjee}
\affil[1]{Department of Electrical and Computer Engineering, University of Washington, Seattle, 98195, USA.}
\affil[2]{Center for Research on Science and Consciousness, Redmond, 98052, USA.}
\affil[3]{Department of Statistics and Data Science, National University of Singapore, 117546, Singapore.}
\affil[*]{rbiswas1@uw.edu}

\usepackage[numbers,compress]{natbib}

\usepackage{booktabs}     
\usepackage{amsfonts}     
\usepackage{nicefrac}      
\usepackage{graphicx}

\usepackage{todonotes}
\usepackage{xcolor}
\usepackage{enumitem}
\usepackage{rotating}
\usepackage{array}
\usepackage{mathtools}

\DeclarePairedDelimiter\floor{\lfloor}{\rfloor}
\usepackage{wrapfig}
\usepackage{url,lineno,microtype,subcaption}
\usepackage[doublespacing]{setspace}
\usepackage{blindtext}
\usepackage[ruled,linesnumbered]{algorithm2e}

\usepackage{bm}
\usepackage{ulem}
\newcommand{\ind}{\perp\!\!\!\!\perp}
\newcommand{\e}{\mathbb{E}}

\usepackage{amsmath,amsthm,amssymb}
\theoremstyle{definition}

\newtheorem{corollary}{Corollary}[section]

\newcommand\numberthis{\addtocounter{equation}{1}\tag{\theequation}}
\newtheorem{theorem}{Theorem}[section]
\newtheorem{lemma}{Lemma}[section]

\makeatletter
 \def\@textbottom{\vskip \z@ \@plus 1pt}
 \let\@texttop\relax
\makeatother
 
\newtheorem{definition}{Definition}[section]

\usepackage{url}

\begin{abstract}
Inferring causation from time series data is of scientific interest in different disciplines, particularly in neural connectomics. While different approaches exist in the literature with parametric modeling assumptions, we focus on a non-parametric model for time series satisfying a Markovian structural causal model with stationary distribution and without concurrent effects. We show that the model structure can be used to its advantage to obtain an elegant algorithm for causal inference from time series based on conditional dependence tests, coined Causal Inference in Time Series (CITS) algorithm. We describe Pearson's partial correlation and Hilbert-Schmidt criterion as candidates for such conditional dependence tests that can be used in CITS for the Gaussian and non-Gaussian settings, respectively. We prove the mathematical guarantee of the CITS algorithm in recovering the true causal graph, under standard mixing conditions on the underlying time series. We also conduct a comparative evaluation of performance of CITS with other existing methodologies in simulated datasets. We then describe the utlity of the methodology in neural connectomics -- in inferring causal functional connectivity from time series of neural activity, and demonstrate its application to a real neurobiological dataset of electro-physiological recordings from the mouse visual cortex recorded by Neuropixel probes.
\end{abstract}
\begin{document}
\maketitle

\section{Introduction}
Inferring causality in time series data is of scientific interest in different disciplines such as neuroscience \citep{reid2019advancing,biswas2022statistical1}, econometrics \citep{gokmenoglu2019time,kwon2011graphical}, climatology \citep{barbero2018temperature} and earth sciences \citep{massmann2021causal}, etc. In neuroscience, it is important to map the causal pathway of neural interaction during brain function. However, the causal neural pathways are not observed and needs to be inferred from recorded neural time series data. The neural interaction pathways are often represented by a directed graph between the neurons termed as the causal functional connectome (CFC) \citep{reid2019advancing, biswas2022statistical1, biswas2022statistical2}, which has an edge from neuron $u\rightarrow v$ if the activity of neuron $u$ at time $t$ has a causal influence on the activity of neuron $v$ at a later time $t'$. Finding the CFC is expected to provide a more fundamental understanding of brain function and dysfunction \citep{hassabis2017neuroscience}. Indeed, finding the CFC would facilitate the inference of the governing causal neuronal interaction pathways essential for brain functioning and behavior \citep{finn2015functional,biswas2022statistical1}. Furthermore, the CFC can be a promising biomarker for neuro-psychiatric diseases. For example, abnormal resting-state functional connectivity between brain regions is also known to predate brain atrophy and the emergence of symptoms of Alzheimer's disease by up to two decades or more \citep{ashraf2015cortical,nakamura2017early, brier2014functional, sheline2013resting}. In other disciplines, such as climate studies, it is of interest to know whether the change in certain climatic variables (e.g. river run-off at a station on a day) has a causal effect on other climatic variables at a later time (e.g. river run-off at a different station on a later day) \citep{molina2017causal,biswas2022statistical2, miersch2023identifying}. In econometrics, it is of interest to know whether the change in one macro-economic variable (e.g. real gross domestic product) at a certain month influences another macro-economic variable (e.g. borrowed reserves) at a later month \citep{geweke1984inference, dahlhaus2003causality, eichler2010granger}.

Methods for inferring causality from time series data often use a parametric model which assumes specific dynamical equations governing the time series variables, such as consisting of a linear vector autoregressive model in popular implementations of Granger Causality \citep{barnett2009granger,granger2001essays, shojaie2022granger}, or non-linear and additive time series model \citep{chu2008search, zaremba2022statistical, nishiyama2011consistent}. In contrast, non-parametric methods do not assume any specific form of the dynamical equations governing the time series variables and thus are widely applicable. Popular non-parametric methods are based on the directed graphical model such as the recently proposed Time-aware PC (TPC) algorithm \citep{biswas2022statistical2, spirtes2000causation}, and the structural causal model (SCM) \citep{peters2013causal,du2007structural}. See Section \ref{sec: cinf_ts_review} for a more detailed review of approaches for causal inference in time series. 

The SCM framework offers researchers with several advantages: 1) it is a graphical approach that is both interpretable and straightforward in verifying causal relationships of interest; 2) the outcomes obtained through the application of SCM do not depend on the specific distributional or functional assumptions that are often assumed in the literature on time series data, such as linear relationships or multivariate normality; 3) the SCM adheres to the Neyman-Rubin causal model \citep{holland1986statistics,rosenbaum1983central, rubin1978bayesian} and earlier works on linear structural equations and causality \citep{hood1953studies,haavelmo1943statistical}. Yet, there has been less attention to leveraging a Markovian condition for the time series causal models. Nevertheless, the Markovian assumption finds wide practical use in time series originating in different disciplines, such as in neuroscience and econometrics \citep{korda2016discrete,pham2017texture, dias2015clustering}. An advantage of the Markovian setting is that the causal information would be present in an interval preceding the current time, which is applicable in practice.  Thereby, in this paper we focus on the non-parametric SCM for time series which is Markovian of an arbitrary but finite order, having a stationary probability distribution and devoid of concurrent effects \citep{peters2017elements,li2017nonparametric}.

It is important to note that causality in time series is expected to follow the arrow of time, i.e. ``causes temporally precede their effects" \citep{eichler2013causal,kleinberg2013causality,lowe2022amortized}. The smallest unit of time that can be recorded by the system is called time granularity \citep{allen1983maintaining,euzenat2005time}. If the time interval between the cause and effect is less than the time granularity, then the cause and effect will appear to be erroneously instantaneous. For many domains of application, selection of proper time granularity is sufficient to prevent such artificially instantaneous causal relationships \citep{wei1982comment, weichwald_jcn}. Even in large databases, representation of time and time granularity are given importance \citep{etzion1998temporal, kulkarni2012temporal, date2014time}. Furthermore, if the interval between consecutive samples of the time series is larger than the time interval between the cause and effect, then the cause and effect can appear to be erroneously concurrent \citep{breitung2002temporal,gong2017causal}. This kind of artificial concurrency can be resolved by more frequent sampling. For instance, in neuroscience, Neuropixels is a recent and popular technology for rapid recording of neural signals in animal models \citep{steinmetz2021neuropixels,steinmetz2018challenges}. Neuropixels record the activity of individual neurons at 30 kHz sampling rate, yielding a fine time granularity of one recording per 0.03 ms. In comparison, transmission of signal from one neuron to another typically happens via neurotransmitters in the inter-neuronal synapse, which has a delay of about 0.5-1 ms for adjacent neurons and longer for non-adjacent neurons \citep{camara2015bio,gabbiani2017mathematics}, precluding artificial concurrency of causal relationships between neurons based on the Neuropixels recordings. In contrast, blood oxygen level dependent (BOLD) signal obtained by functional magnetic resonance imaging (fMRI) is expected to include artificially concurrent causal relationships between regions of interest, since the transformation of neural activity to haemodynamic response in BOLD signal temporally aggregates neural activity over several seconds, while the time lag of causal neural interactions is of the order of milliseconds \citep{buxton1998dynamics,glover1999deconvolution,bush2013decoding}.

While the recently proposed TPC algorithm computes causality from time-series data where concurrent causal relationships may be present \citep{biswas2022statistical2}, in this paper, we present a method to compute causality with greater efficacy from time series data which does not involve concurrent causality. Such a greater efficacy is due to the following: unlike in TPC, the Directed Markov Property of the time series windows is no longer an assumption in this paper, and instead, the same property for the entire time series follows naturally from the more interpretable time series structural causal model \eqref{eqdef: process}. While consistent causal inference by the TPC algorithm was demonstrated in \citep{biswas2022consistent} for popular time series models such as vector auto-regressive moving average and linear processes, the framework of the current paper enables us to prove consistent causal inference under the more general model \eqref{eqdef: process} on the underlying time series, permitting arbitrary functional relationships between causal antecedents and precedents in the time series.

For a Markovian time series of order $\tau$, the direct causes of the variables at a time $t$ can be variables at times not preceding $t-\tau$, and we show that the conditional independence information in a $2\tau$ window before time $t$ can be used to detect the direct causes of variables at time $t$. We then use this property to our advantage to propose an elegant algorithm for causal inference in this time series setting, coined the Causal Inference in Time Series (CITS) algorithm. We show the mathematical guarantee of the algorithm in yielding the correct DAG under standard mixing conditions on the underlying time series. Tests based on partial correlation and the Hilbert-Schmidt independence criterion are candidates for such statistical tests for conditional dependence in the time series setting \citep{biswas2022consistent}. We demonstrate the performance of the algorithm in simulated datasets. We then describe the utility of the methodology in neural connectomics to infer causal functional connectivity between neurons from time series of neural activity. We apply the methodology to a dataset of neural signals recorded by Neuropixels to obtain the causal functional connectivity between neurons in the mouse visual cortex.

\section{Results}
\subsection{Time series structural causal model}\label{sec:tsscm}
In the time series setting we consider, the data consists of a finite realization of a strictly stationary multivariate Markovian process $\{\mathbf{X}_t\}_{t\in \mathbb{Z}}$ of order $\tau$ and having $p$ components, i.e. $\mathbf{X}_t = (X_{1,t},\ldots,X_{p,t})$ for every $t\in \mathbb{Z}$. The number of components $p$ is arbitrary but fixed. There can be serial dependence in the sequence between variables within the same as well as between different components. We also assume that the stochastic process satisfies a structural causal model (SCM), which remains invariant across time $t$. The SCM consists of a collection of assignments: \begin{equation}\label{eqdef: process}X_{v,t}=f_{v}(X_{\text{pa}(v,t)},\epsilon_{v,t}),~ v=1,\ldots,p, t\in \mathbb{Z},\end{equation}
where, $\text{pa}(v,t)\subseteq \{(d,s): d=1,\ldots,p; s=t-1,\ldots,t-\tau\}$, $\epsilon_{v,t}, v=1,\ldots,p, t\in \mathbb{Z}$ are jointly independent, and, for $S\subseteq \{1,\ldots,p\}\times \mathbb{Z}$, $X_S:= (X_{\bm{i}})_{\bm{i}\in S}$. 

By definition \eqref{eqdef: process}, we are not allowing concurrent effects (see \citep{heise1970causal,eichler2013causal,li2017nonparametric}), i.e. $(u,t)\not\in\text{pa}(v,t)$, for $u,v=1,\ldots,p, t\in \mathbb{Z}$.

The graph $G$ of the SCM is obtained by creating one vertex for each $(v,t),v=1,\ldots,p; t=1,\ldots,\tau+1$ and drawing directed edges from each parent in $\text{pa}(v,\tau+1)$ to $(v,\tau+1)$ for $v=1,\ldots,p$. Sometimes, by a slight misuse of notation, we will refer to a vertex $(v,t)$ by $X_{v,t}$. We assume that this graph is acyclic, and then $G$ is referred to as a Directed Acyclic Graph (DAG). An example is given in Fig \ref{fig:graphexample} (left).

\begin{figure}
    \centering
    \includegraphics[width = 0.5\textwidth]{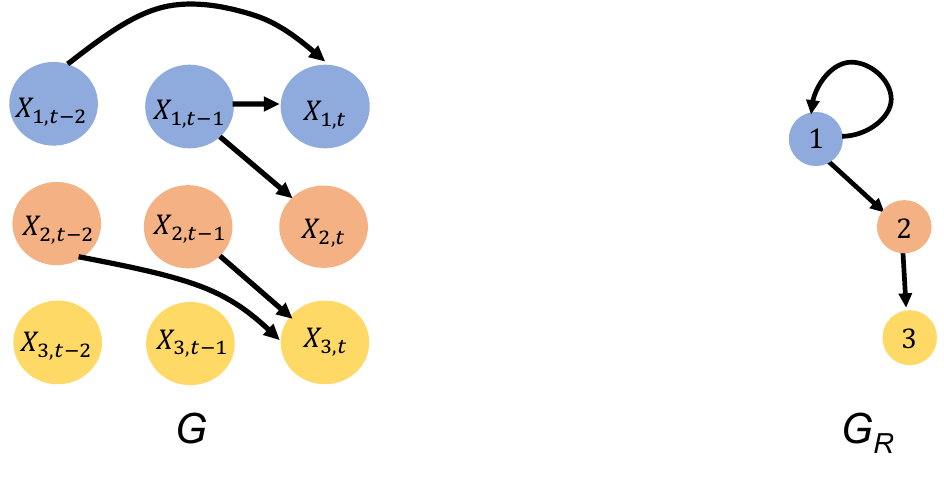}
    \caption{An example of the time-invariant DAG $G$ (left) and its Rolled graph $G_R$ (right) (See Definition \ref{def:rolledgraph}) corresponding to a stationary Markovian time series of order 2 satisfying the SCM: $X_{1,t} = X_{1,t-1} + 0.5 X_{1,t-2} + \epsilon_{1,t}, ~ X_{2,t} = \sin(X_{1,t-1}) + \epsilon_{2,t}$ and $X_{3,t} = \sqrt{\vert X_{2,t-2} + X_{2,t-1}\vert} + \epsilon_{3,t}$.}
    \label{fig:graphexample}
\end{figure}

Due to stationarity, the DAG $G$ of the SCM is invariant over time $t$, that is, for any time $t\geq \tau+1$, the DAG of the edges from $X_{\text{pa}(v,t)}$ to $X_{v,t}$ for $v=1,\ldots,p$ is $G$. Due to the Markovian structure, $\text{pa}(v,t)$ can have nodes at time points $t,t-1,\ldots,t-\tau$ and no earlier. Furthermore, due to no concurrent effects, time $t$ is excluded from $\text{pa}(v,t)$, since the edges are directed forward in time, and there are no directed edges across the different variables at the same time point (See Fig \ref{fig:graphexample}). We also assume that the SCM satisfies faithfulness. Then, the DAG $G$ can be identified up to its Markov Equivalence Class from the distribution of the stochastic process. In fact, in this scenario, the DAG $G$ is identified uniquely by virtue of temporal order, since all graphs in the Markov Equivalence Class have the same skeleton among which the DAG $G$ is the only member that obeys the time order.

The DAG $G$ with nodes corresponding to variable $v$ at time $t$ and edges representing the causal influence of variable $u$ at time $t_1$ to variable $v$ at time $t_2$ is also referred to as the \textit{unrolled DAG} for the time series. The unrolled DAG for the time series captures the across-time causal relationships between the variables over time. The unrolled DAG is often summarized by a \textit{rolled graph} with nodes as the variables and an edge $u\rightarrow v$ for variables $u, v$ if variable $u$ at some time $t$ has a causal effect on variable $v$ at a later time $t'$ (see Figure \ref{fig:graphexample}) \citep{peters2013causal,biswas2022statistical2}. The rolled graph obtained from $G$ summarizes the interaction of the variables over time. Therefore, an alternative name used for rolled graph in literature is \textit{summary causal graph} for time series \citep{peters2013causal,assaad2022entropy}. In the context of neural connectomics, the rolled graph between neurons obtained from neural time series represents the CFC between the neurons since the CFC has an edge from neuron $u\rightarrow v$ if neuron $u$ at some time $t$ has a causal influence on neuron $v$ at a later time $t'$ \citep{reid2019advancing,valdes2011effective,biswas2022statistical2}.

\begin{definition}\label{def:rolledgraph}
The Rolled Graph of $G$, denoted by $G_R$, is defined as a graph with nodes $1,\ldots,p$ which has an edge $u\rightarrow v$ if $(u,t_1)\rightarrow (v,t_2)$ is an edge in the unrolled DAG $G$ for some $t_1\leq t_2$. 
\end{definition}

Since we assume no concurrent effect, it suffices to look for edges in $G$ of the form $(u,t_1)\rightarrow (v,t_2)$ for some $t_1< t_2$. Furthermore, due to stationarity and Markovian property, it suffices to look for edges incident on an arbitrary fixed time point $t$ and originating from $t-\tau,\ldots,t-1$. Thereby, $G_R$ has an edge $u\rightarrow v$ if and only if $X_{u,s}\rightarrow X_{v,2\tau+1}$ for some $s=\tau+1,\ldots,2\tau$. In the context of neuroscience, the Rolled Graph $G_R$ is termed as causal functional connectivity between neurons.

\subsection{Model properties facilitating causal inference}\label{sec:properties}
The goal of this paper is to estimate $G$ and thereafter $G_R$ in a non-parametric manner without relying on particular model specifications for the underlying time series. We illustrate the goal with a simple example of a stationary VAR(p)-model: $X_{v,t} = \sum_{u=1}^p \sum_{j=1}^{\tau} \phi_{uv}^{(j)} X_{u,t-j}+\epsilon_{v,t}$, where $\tau$ is the order of the Markovian process, and the noise variables $\epsilon_{v,t}$ are i.i.d. with mean zero and $\epsilon_{1,t},\ldots,\epsilon_{p,t}$ are independent of $\{X_{u,s}:u=1,\ldots,p, s<t\}$. Then, the entry of the adjacency matrix of $G$ corresponding to the edge $(u,t-j)\rightarrow (v,t)$ in $G$ is, \begin{equation}\label{eq:var_adj}\mathbf{1}(\phi_{uv}^{(j)}\neq 0).\end{equation}

In this scenario, the weights $\phi_{uv}^{(j)}$ can be estimated by a Likelihood Ratio (LR) test assuming Gaussian distributed noise terms, and plugging them into \eqref{eq:var_adj}, one can estimate the adjacency matrix of $G$ and therefore $G$. This approach of estimating the weights in the VAR model and using them to estimate $G$ is essentially the \textit{Granger Causality}. However, if the noise distribution is unknown then such an LR test is difficult to be conducted. Furthermore, if the true underlying data generating stationary process is nonlinear, perhaps with non-additive innovation terms, it is non-trivial to extend this approach. 

Instead, we use a novel approach based on conditional independence (CI) which is motivated from other CI-based methods such as the SGS \citep{glymour1991causal} and PC \citep{kalisch2007estimating} algorithms, and utilizes the time order, Markovian and stationary distribution to its advantage to infer the DAG in the time series setting. The adjacency of edges in $G$ can be determined from a conditional independence oracle because of the following: $X_{v,t}$ and $X_{u,s}$ are conditionally independent given a set of other nodes if and only if $X_{v,t}$ and $X_{u,s}$ are d-separated in $G$ by that conditioning set of nodes under faithfulness, which once again holds if and only if $X_{v,t}$ and $X_{u,s}$ are non-adjacent \citep{verma2022equivalence}. Furthermore, $X_{v,t}$ and $X_{u,s}$ are non-adjacent if and only if $X_{v,t}$ and $X_{u,s}$ are d-separated by their parents \citep{verma2022equivalence}. It is noteworthy that the parents of $X_{v,t}$ and $X_{u,s}$ are within the interval $\{t\wedge s-\tau, \ldots, t\vee s-1\}$ since the process is Markovian for order $\tau$. Therefore, $X_{v,t}$ and $X_{u,s}$ are conditionally independent given a set of other nodes in the interval $\{t\wedge s-\tau, \ldots, t\vee s-1\}$ if and only if $X_{v,t}$ and $X_{u,s}$ are non-adjacent.

This illustrates that we can relate the adjacency of a pair of nodes $X_{v,t}$ and $X_{u,s}$ in $G$ to conditional independence information of the pair of nodes given a set of other nodes in the interval $\{t\wedge s-\tau, \ldots, t\vee s-1\}$. This is formalized by the following lemma, proved in Appendix \ref{proofLem3.1}.

\begin{lemma}\label{lemma:concept}
    For $u,v=1,\ldots,p$ and $t\in \mathbb{Z}$, $s\in \{t-\tau,\ldots,t-1\}$, the following statements are equivalent.
    \begin{enumerate}
        \item [(1)] $X_{u,s}\not\in \text{pa}(v,t)$.
        \item [(2)] $X_{v,t}$ and $X_{u,s}$ are non-adjacent in $G$.
        \item [(3)] $X_{v,t}\ind X_{u,s} \vert \bm{X}_{S}$ 
    for some $S\subseteq \{(d,r):d=1,\ldots,p; r= t-2\tau,\ldots,t-1\}$ and $\bm{X}_{S} = \{X_{d,r}: (d,r)\in S\}$.
    \end{enumerate}
\end{lemma}

We recall that all elements of $\text{pa}(v,t)$ are within times $t-\tau,\ldots,t-1$,since the time series is Markovian of order $\tau$. Therefore it follows from Lemma \ref{lemma:concept} that elements of $\text{pa}(v,t)$ are $X_{u,s}$ for $u=1,\ldots,p; s\in \{t-\tau,\ldots,t-1\}$ such that \begin{equation}\label{eq:prop_adj}X_{u,s}\not\ind X_{v,t} \vert \bm{X}_{S} \text{ for some } S\subseteq \{(d,r):d=1,\ldots,p; r= t-2\tau,\ldots,t-1\}.\end{equation}

We use this key property to formulate the Causal Inference in Time Series (CITS) algorithm to estimate $G$ and $G_R$ in the following section and show its mathematical guarantee in estimation later in the paper. 


\subsection{Inference from data: The CITS algorithm} 
\subsubsection{Oracle version: With a Conditional Dependence Oracle}
The properties described earlier (see Section \ref{sec:properties}, Lemma \ref{lemma:concept}) provide us a strategy to obtain $\text{pa}(v,t)$ by two levels of searching: 1. searching for $X_{u,s}$, $u=1,\ldots,p; s\in \{t-\tau,\ldots,t-1\}$ that satisfy property \eqref{eq:prop_adj}, and 2. searching for subsets $S$ of $\{(d,r):d=1,\ldots,p; r= t-2\tau,\ldots,t-1\}$ to determine if property \eqref{eq:prop_adj} is satisfied for at least one such subset. Note that to verify property \eqref{eq:prop_adj}, one can stop the search when one such subset $S$ is found, whereas to conclude that property \eqref{eq:prop_adj} is not satisfied, one has to check all such subsets. In particular, using this we can obtain whether $X_{v,t}$ and $X_{u,s}$ are adjacent and thereby for $s\in \{t-\tau,\ldots,t-1\}$, $X_{u,s}\rightarrow X_{v,t}$ i.e. $X_{u,s}\in \text{pa}(v,t)$ by Lemma \ref{lemma:concept}. In particular, we can obtain $\text{pa}(v,2\tau+1)$ for all $v=1,\ldots,p$. Next, $G$ is obtained by assigning edges from $\text{pa}(v,2\tau+1)$ to $X_{v,2\tau+1}, ~ v=1,\ldots,p$, and $G_R$ is the Rolled Graph of $G$ with nodes $1,\ldots, p$ and edge $u\rightarrow v$ for $u,v=1,\ldots,p$ if $X_{u,s}\in pa(v,2\tau+1)$ for some $s=\tau+1,\ldots,2\tau$. These steps constitute the Time-series Structural Causal Inference algorithm oracle version (CITS-oracle). The oracle has knowledge of conditional independence information between the variables. The steps are clearly outlined in Algorithm \ref{alg:CITS_oracle}. The following theorem establishes that the CITS-Oracle outputs the true DAG for the stochastic process and its Rolled Graph.

\begin{theorem}\label{thm:CITS_oracle}
    Let $\{\bm{X}_t\}_{t\in \mathbb{Z}}$ be a strictly stationary Markovian process of order $\tau$, and assume that it can be uniquely represented in the form of a structural causal model with time invariant DAG $G$ as in \eqref{eqdef: process}, which exhibits no concurrent effects and satisfies faithfulness. Then, the CITS-Oracle algorithm outputs the true DAG $G$ and its Rolled Graph $G_R$.
\end{theorem}
\begin{proof}
    Since the DAG is time invariant, in order to obtain $G$, it suffices to obtain parental sets of variables $\text{pa}(v,2\tau+1)$ at a fixed time $t=2\tau+1$ and $v=1,\ldots,p$. This justifies to inputting times $t=1,\ldots,2\tau+1$ to the algorithm \ref{alg:CITS_oracle} to obtain $G$.

    Line 1 of algorithm \ref{alg:CITS_oracle} initializes the DAG. For any $u,v=1,\ldots,p$, if there is an edge in $G$ from $X_{u,s}$ to $X_{v,t}$, then $s<t$ holds, due to time order. This justifies having only edges $X_{u,s}\rightarrow X_{v,t}$ with $s<t$ for the initialization.
    
    Lines 2-9 consider each possible edge $X_{u,s}\rightarrow X_{v,2\tau+1}$ and deletes the edge if for some $S\subseteq\{(d,r):d=1,\ldots,p; r=1,\ldots,2\tau+1\}$, $X_{u,s}\ind X_{v,2\tau+1}\vert \bm{X}_S$, where $\bm{X}_{S} = \{X_{d,r}: (d,r)\in S\}$. This is justified by the implication (3)$\implies$ (1) in Lemma \ref{lemma:concept} concluding that $X_{u,s}\not\in \text{pa}(v,2\tau+1)$.
    
    In fact, by Lemma \ref{lemma:concept} implication (1) $\implies$ (3), for the remaining edges $X_{u,s}\rightarrow X_{v,2\tau+1}$, $X_{u,s}$ will be the parents of $X_{v,2\tau+1}$, $u=1,\ldots, p$, $s=\tau+1,\ldots,2\tau+1$. This justifies that line 10 correctly finds the parental set $\text{pa}(v,2\tau+1)$ from the remaining edges in $G$.

    Based on the parental sets $\text{pa}(v,2\tau+1)$, $v=1,\ldots,p$, line 11 trivially directs the edges from $\text{pa}(v,2\tau+1)$ to $X_{v,2\tau+1}$ to outputs the DAG $G$ and line 12 converts $G$ to its Rolled Graph $G_R$. 

\end{proof}

\begin{algorithm}[t!]
\SetKwInOut{Input}{Input}
\SetKwInOut{Output}{Output}
\Input{$X_{v,t}, v=1,\ldots,p; t=1,\ldots,2\tau+1$, Conditional Independence Information.}
\Output{DAG $G$ and Rolled DAG $G_R$}

Start with an initial DAG $G_1$ between nodes $\{X_{v,t}: v=1,\ldots,p; t=1,\ldots,2\tau+1\}$ with edges $X_{u,s}\rightarrow X_{v,2\tau+1}$ for $s=\tau+1,\ldots,2\tau$, $u,v=1,\ldots,p$.

\Repeat{all $u,v=1,\ldots,p$, $s=\tau+1,\ldots,2\tau+1$ are tested.}{
\Repeat{edge $X_{u,s}\rightarrow X_{v,2\tau+1}$ is deleted or all $S\subseteq\{(d,r):d=1,\ldots,p; r=1,\ldots,2\tau+1\}\setminus\{(u,s),(v,2\tau+1)\}$ are selected.}{
            Choose $S\subseteq\{(d,r):d=1,\ldots,p; r=1,\ldots,2\tau+1\}$.
            
            \uIf{$X_{u,s}\ind X_{v,2\tau+1} ~\vert~ \bm{X}_{S}$}{
                
                Delete edge $X_{u,s}\rightarrow X_{v,t}$.
                
                Denote this new graph by $G_1$.
}
}
}

$\text{pa}(v,2\tau+1) = \{X_{u,s}: X_{u,s}\rightarrow X_{v,2\tau+1} \text{ in $G_1$}; s=\tau+1,\ldots,2\tau; u=1,\ldots,p\}$.

Obtain the DAG $G$ by edges directing from $\text{pa}(v,2\tau+1)\rightarrow X_{v,2\tau+1}$.

Obtain the Rolled Graph $G_R$ with nodes $v=1,\ldots,p$ and edge $u\rightarrow v$ if $X_{u,s}\in \text{pa}(v,2\tau+1)$ for some $u=1,\ldots,p; s=\tau+1,\ldots,2\tau$.

\caption{CITS-Oracle}\label{alg:CITS_oracle}
\end{algorithm}

\subsubsection{Sample version: Without a Conditional Dependence Oracle}
For the sample version of the CITS algorithm (CITS-sample), we replace the conditional independence statements by outcomes of statistical tests for conditional dependence based on a sample. For details of appropriate conditional dependence tests, see Appendix \ref{appsec:condindtests}. It is noteworthy that the observed data does not have replicates of the entire stochastic process and only consists of a single realization of the stochastic process up to time $n$. Considering this, we first construct time-windowed samples by taking consecutive time windows of a duration of $2\tau+1$. That is, the samples are $\chi_k = \{X_{v,t}:v=1,\ldots,p; t= (2\tau+1)(k-1)+1,\ldots,(2\tau+1)k\}, k=1,\ldots,N$ where $N=\floor{\frac{n}{2\tau+1}}$. For example, for some $v=1,\ldots,p$ and $t=1,\ldots,2\tau+1$, the samples for $X_{v,t}$ based on $\{\chi_k\}_{k=1}^N$ are: $\{X_{v,(2\tau+1)(k-1)+t}:k=1,\ldots,N\}$. Then for $u,v=1,\ldots,p; t=2\tau+1; s=\tau+1,\ldots,2\tau+1; S\subset\{(d,r):d=1,\ldots,p; r= 1,\ldots,2\tau+1\}$, we test the hypotheses $X_{u,s}\ind X_{v,t} \vert \bm{X}_{S}$ based on samples $\chi_k$. We then estimate $\text{pa}(v,2\tau+1)$ using the same steps as CITS-oracle, but replacing the conditional independence statements by the outcome of the statistical tests. The CITS-sample algorithm is outlined in Algorithm \ref{alg:CITS_sample}.

\begin{algorithm}[t!]
\SetKwInOut{Input}{Input}
\SetKwInOut{Output}{Output}
\Input{$X_{v,t}, v=1,\ldots,p; t=1,\ldots,n$}
\Output{Estimated DAG $\hat{G}$ and Rolled Graph $\hat{G}_R$}

Construct Time-Windowed samples: $\chi_k = \{X_{v,t}:v=1,\ldots,p; t= (2\tau+1)(k-1)+1,\ldots,(2\tau+1)k\}, k=1,\ldots,N$.

Run the CITS-Oracle algorithm while replacing the conditional independence statement in Line 5 by statistical tests in Appendix \ref{appsec:condindtests} based on samples $\chi_k$ to output DAG $\hat{G}$ and Rolled Graph $\hat{G}_R$.

\caption{CITS-Sample}\label{alg:CITS_sample}
\end{algorithm}

\subsection{Comparative study of performance in simulated datasets}
We compare the performance of CITS, Pairwise Granger Causality (GC1), Multivariate Granger Causality (GC2), naive application of PC algorithm (PC), and Time-Aware PC algorithm (TPC), to recover the ground truth causal relations in simulated datasets. We use simulated datasets from a variety of time series models ranging from linear to non-linear models, with and without common causes, and consider both the Gaussian and non-Gaussian noise settings (See Appendix \ref{appen:simulstudy}).

In the simulations, for each model, we obtain 25 simulations of the entire time series each for different noise levels $\eta \in \{0.1,0.5,1,1.5,2,2.5,3,3.5\}$. All the time series simulated have $n=1000$ samples. We also used the level $\alpha$ of the conditional dependence test with $\alpha$ ranging in $0.01, 0.05$ and $0.1$. The performance of the methods in recovering the ground truth causal relationships is summarized using the following three metrics: (1) Combined Score (CS), (2) True Positive Rate (TPR) and (3) 1 - False Positive Rate (IFPR). Let True Positive (TP) represent the number of correctly identified edges, True Negative (TN) represent the number of correctly identified missing edges, False Positive (FP) represent the number of incorrectly identified edges, and False Negative (FN) represent the number of incorrectly identified missing edges across simulations. IFPR is defined as: $$\text{IFPR}=\left(1-\frac{\text{FP}}{\text{FP+TN}}\right)\cdot 100,$$ which measures the ratio of the number of correctly identified missing edges by the algorithm to the total number of true missing edges. Note that the rate is reported such that $100\%$ corresponds to no falsely detected edges. TPR is defined as $\text{TPR}=\left(\frac{\text{TP}}{\text{TP} + \text{FN}}\right)\cdot 100$ i.e. the ratio of the number of correctly identified edges by the algorithm to the total number of true edges in percent. The Combined Score (CS) is given by Youden's Index \citep{vsimundic2009measures,hilden1996regret}, as follows, $\text{CS} = \text{TPR} - \text{FPR}$. 

In the Gaussian settings, we use the partial correlation based conditional dependence test and in the non-Gaussian settings we use the Hilbert-Schmidt Criterion. The Gaussian conditional dependence tests, in practice, use a fixed level $\alpha$ and takes the form $\sqrt{n-k-3} ~ \log(\frac{1+\hat{\rho}}{1-\hat{\rho}}) \leq \Phi^{-1}(1-\alpha)$, similar to the \textit{pcalg} package in \textit{R} and other software such as \textit{TETRAD$^{IV}$}. One can thus use $\gamma = \Phi^{-1}(1-\alpha)/\sqrt{n-\bm{k}-3}$, which also gives rise to a consistent test.


In each of our simulation settings, there are $4$ neurons and $16$ possible edges (including self-loops), leading to a total of $400$ possible edges across $25$ simulations. In Fig \ref{fig:compsim}, we compare in detail the results for GC1, GC2, PC, TPC and CITS for inferring the true Rolled graph at noise level $\eta =1$ and thresholding parameter $\alpha = 0.05$. In Fig \ref{fig:compvaralph}, we compare the Combined Score (CS) of the approaches over different noise levels $\eta$ and significance levels $\alpha$ for each simulation setting.

\begin{figure}
    \centering
    \includegraphics[width=0.4\textwidth]{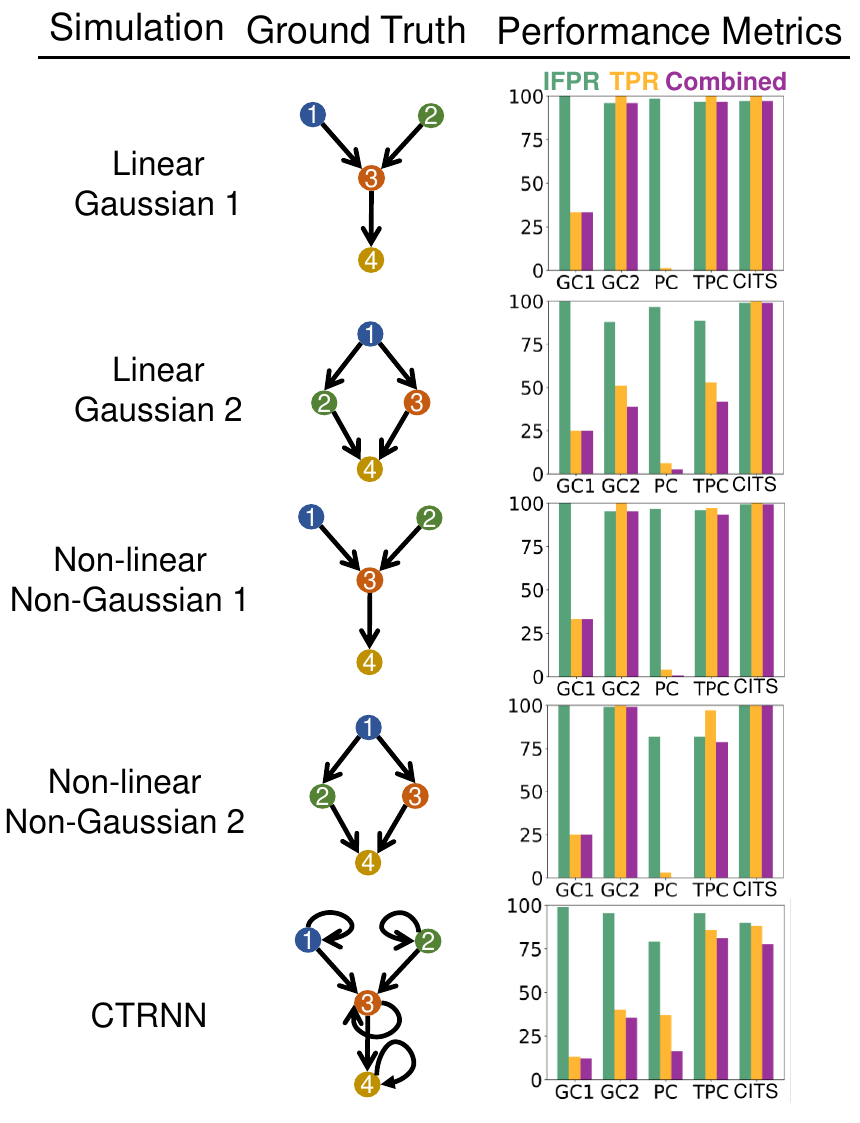}
    \caption{\textbf{Comparative study of inferring the Rolled graph.} Inference of Rolled graph for five simulation settings (top to bottom in left column): Linear Gaussian Models 1 and 2, Non-linear Non Gaussian Models 1 and 2, CTRNN. The ground truth for each simulation paradigm is graphically represented (middle column). The performances of the five methods: GC1, GC2, PC, TPC and CITS, are shown in terms of three metrics (right column): IFPR (green), TPR (orange) and CS (purple). }
    \label{fig:compsim}
\end{figure}
\begin{figure}
    \centering    
    \includegraphics[width=\textwidth]{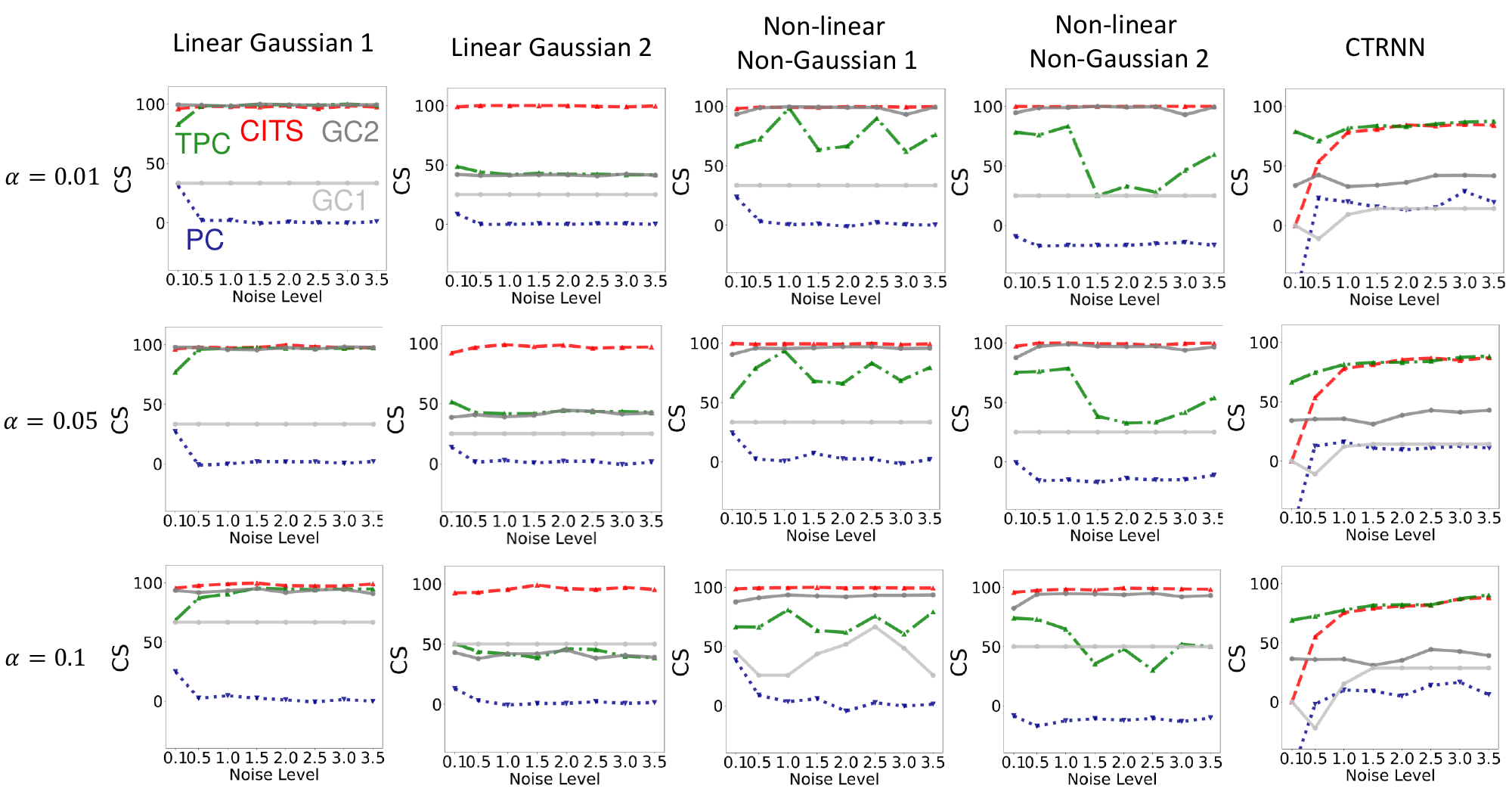}
    \caption{\textbf{Comparison of performance across different noise levels $\eta$ and significance levels $\alpha$.} Combined Score of the four methods of causal inference of Rolled graph - CITS (red), TPC (green), PC (blue), GC1 (light gray), GC2 (dark gray) over varying noise levels in simulation $\eta = 0.1,0.5,1.0,\ldots,3.5$, for simulated motifs from Linear Gaussian, Non-linear Non-Gaussian and CTRNN paradigms (left to right), with significance level $\alpha = 0.01, 0.05, 0.1$ (top to bottom).}
    \label{fig:compvaralph}
\end{figure}

\begin{itemize}[leftmargin=0pt]
    \item[] In the \textit{Linear Gaussian 1 scenario (top row in Fig \ref{fig:compsim})}, the Ground Truth Rolled graph consists of $1\rightarrow 3$, $2\rightarrow 3$ and $3\rightarrow 4$. Among the methods for inference, GC1, GC2, PC, TPC and CITS produce $\text{IFPR} = 100\%, 96\%, 98.5\%, 96.6\%, 97.2\%$, $\text{TPR}=33.3\%, 100\%, 1.3\%, 100\%$, and $100 \% $, and CS $= 33.3\%, 96\%, -0.2\%, 96.6\%, 97.2\%$ respectively. Therefore, among these methods, we conclude that GC1 has very low false positive rate, and hence, is highly specific to detecting the true edges. However, it has low true positive rate, hence, it misses out detecting many of the true edges. Among the approaches, CITS has the highest performance that is closely followed by TPC and GC2, and then GC1, and lastly PC, which has a poor performance.
    
    \item[] In the \textit{Linear Gaussian 2 scenario (second row)}, the Ground Truth Rolled graph consists of edges $1\rightarrow 2$, $1\rightarrow 3$, $2\rightarrow 4$, $3\rightarrow 4$, which comprise a common cause in vertex $1$ and a common effect in vertex $4$. Among the methods for inference, GC1, GC2, PC, TPC and CITS produce $\text{IFPR} = 100\%, 88\%, 96.7\%, 88.7\%, 99\%$, $\text{TPR}=25\%, 51\%, 6\%, 53\%$, and $100 \% $, and CS $= 25\%, 39\%, 2.7\%, 41.7\%, 99\%$ respectively. Therefore, among the methods, CITS has the best performance in CS, followed by TPC, GC2, GC1 and lastly PC. The poor performance of GC2 in this setting can be explained by multi-collinearity in regression due to presence of a common cause.
    \item[] In the \textit{Non-linear Non-Gaussian 1 scenario (third row)}, the Ground Truth CFC consists of edges $1\rightarrow 3, 3\rightarrow 4$ as causal influences of an increasing nature due to $\sin(x)$ being an increasing function in the domain $[0,\frac{\pi}{2}]$, while the edge $2\rightarrow 3$ is a decreasing influence due to $-\sin(x)$ being a decreasing function within the same domain. This setting also comprises a common effect in vertex $3$ but no common cause. GC1, GC2, PC, TPC and CITS yielded $\text{IFPR} = 100\%, 95.4\%, 96.6\%, 96\%, 99.4\%$ and $\text{TPR} = 33.3\%, 100\%, 4\%, 97.3\%, 100\%$ and CS $=33.3\%, 95.4\%, 0.6\%, 93.3\%$ and $99.4\%$. In terms of CS, CITS has the highest performance among the methods.
    \item[] In the \textit{Non-linear Non-Gaussian 2 scenario (fourth row)}, the Ground Truth CFC consists of edges $1\rightarrow 2, 1\rightarrow 3, 2\rightarrow 4, 3\rightarrow 4$, which comprise a common cause in vertex $1$ and a common effect in vertex $4$. GC1, GC2, PC, TPC and CITS yielded $\text{IFPR} = 100\%, 99\%, 81.7\%, 81.7\%, 100\%$ and $\text{TPR} = 25\%, 100\%, 3\%, 97\%, 100\%$ and CS $=25\%, 99\%, -15.3\%, 78.7\%$ and $100\%$. In terms of CS, CITS again has the highest performance among the methods. 
    \item[] In \textit{CTRNN scenario (third column)}, self-loops are present for each vertex. It is noteworthy that GC1 and PC does not obtain self-loops. Among the methods, IFPR of GC1, GC2, PC, TPC and CITS is $99.1\%, 95.6\%, 79.1\%, 95.6\%, 89.8\%$ and $\text{TPR}$ is $13.1\%, 40\%, 37.1\%, 85.7\%, 88\%$ and CS is $12.3\%, 35.6\%, 16.3\%, 81.3\%, 77.8\%$ respectively. Among all methods, CITS has the highest TPR, followed by TPC. GC1 has the highest IFPR, followed by GC2 and TPC, then CITS and lastly PC. In terms of the CS, TPC has the highest performance, closely followed by CITS, compared to other methods, with a CS of $42.2\%$ higher than the next best method.
\end{itemize}

We compare the Combined Score of CITS with other approaches across noise levels $\eta$ ranging from $0.1$ to $3.5$ and significance levels $\alpha = 0.01,0.05,0.1$ in Fig \ref{fig:compvaralph}. In the Linear Gaussian 1 scenario, CITS has a CS of nearly $100\%$ across noise levels greater than $1.0$ and all signifiance levels considered, closely matching the parametric GC2 model as well as non-parametric TPC, which are followed by GC1 in performance and lastly PC. In the Linear Gaussian 2 scenario, the distinction is remarkable, where CITS has a CS of $\approx 100\%$ across all levels of simulation noise and significance level $\alpha$, however the next best model for this setting are TPC and GC2 with a CS of $\approx 50\%$, followed by GC1 and lastly PC. In the Non-linear Non-Gaussian 1 and 2 scenarios, CITS has the highest CS compared to other methods across levels of noise and $\alpha$. In the CTRNN scenario, the best CS achieved is lower compared to other simulation paradigms. However, CITS and TPC have better performance compared to the other methods for noise level $\eta\geq 0.5$ and all $\alpha$.

\begin{figure}[t]
    \centering
    \includegraphics[width = \textwidth]{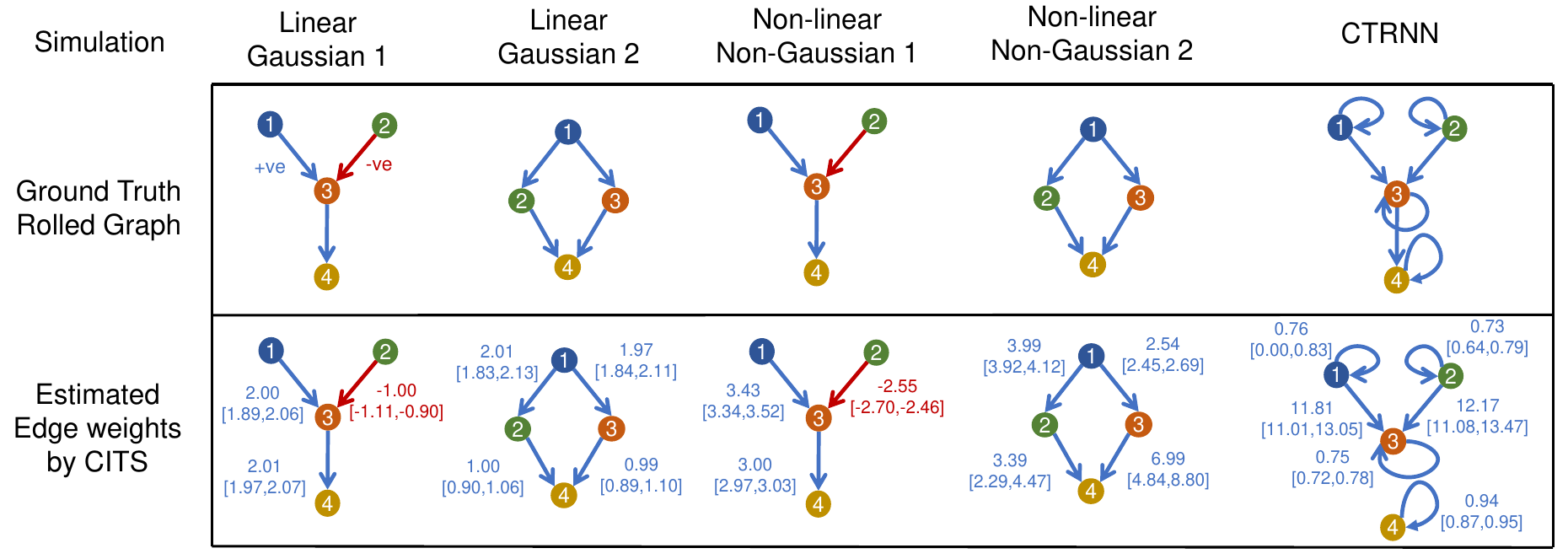}
    \caption{Estimated edge weights (median [min, max])  from Section \ref{sec:edgeweights} for the simulation paradigms: Linear Gaussian 1, Linear Gaussian 2, Non-linear Non-Gaussian 1, Non-linear Non-Gaussian 2, and CTRNN (left to right).}
    \label{fig:edge_wts}
\end{figure}

In Fig \ref{fig:edge_wts}, we present the results on edge weights and inferred nature of edges (whether they are positive/increasing or negative/decreasing) for noise level $\eta = 1$ and thresholding parameter $\alpha = 0.05$. For the \textit{Linear Gaussian 1 scenario}, the estimated edge weights for connections $1 \rightarrow 3$, $2 \rightarrow 3$, and $3 \rightarrow 4$ are 2.00, -1.00, and 2.01, respectively, in median across simulation trials, and range between [1.89, 2.06], [-1.11,-0.90] and [1.97, 2.07] respectively. These estimated weights thus closely agree with the true linear weights 2, -1 and 2 respectively in the Ground Truth model. It is noteworthy that the positive or negative nature of the estimated edges across all simulation trials are in agreement with the Ground Truth. 

In the \textit{Linear Gaussian 2 scenario}, the estimated edge weights for connections $1 \rightarrow 2$, $1 \rightarrow 3$, and $2 \rightarrow 4$ and $3\rightarrow 4$ are 2.01, 1.97, 1.00 and 0.99, respectively, in median across simulation trials, and range between [1.83,2.13], [1.84,2.11], [0.90,1.06], and [0.89,1.10] respectively. Again, the estimated values of the edge weights are in close agreement with the linear weights 2, 2, 1, 1 respectively in the Ground Truth model, and the estimated positive nature of the edges are in agreement with that of the Ground Truth in all the simulation trials. 

In the \textit{Non-linear Non-Gaussian 1} scenario, the estimated edge weights for connections $1 \rightarrow 3$, $2 \rightarrow 3$, and $3 \rightarrow 4$ are 3.43, -2.55, and 3.00, respectively, in median across simulation trials, and range between [3.34, 3.52], [-2.70,-2.46] and [2.97, 3.03] respectively. The signs of the estimated edges across all simulation trials are thereby in agreement with the Ground Truth. 

In the \textit{Non-linear Non-Gaussian 2} scenario, the estimated edge weights for connections $1 \rightarrow 2$, $1 \rightarrow 3$, and $2 \rightarrow 4$ and $3\rightarrow 4$ are 3.99, 2.54, 3.39 and 6.99, respectively, in median across simulation trials, and range between [3.92,4.12], [2.45,2.69], [2.29,4.47], and [4.84,8.80] respectively. Again, the estimated positive nature of the edges are in agreement with that of the Ground Truth in all the simulation trials. 

Lastly, in the \textit{CTRNN scenario}, the estimated edge weights for self-connections $1\rightarrow 1$, $2\rightarrow 2$, $3\rightarrow 3$ and $4\rightarrow 4$ are 0.76, 0.73, 0.75 and 0.94, respectively, in median across all simulation trials, and range between [0, 0.83], [0.64, 0.79], [0.72, 0.78] and [0.87, 0.95] respectively. The estimated edge weights for the non-self edges $1 \rightarrow 3$ and $2 \rightarrow 3$ are 11.81 and 12.17, respectively, in median across all simulation trials, and range between [11.01, 13.05] and [11.08, 13.47] respectively. Here also, the estimated positive nature of the edges is in agreement with the Ground Truth in all the simulation trials.

\subsection{Application to neural connectomics: Neuropixels data}\label{sec:application}

In neuroscience, the term Functional Connectome (FC) refers to the network of interactions between individual units of the brain, such as neurons or brain regions, with respect to their activity over time \citep{reid2012functional, biswas2022statistical1}. The main purpose of identifying the FC is to gain an understanding of how neurons work together to create brain function. The FC can be represented as a graph, where nodes denote neurons and edges denote a stochastic relationship between the activities of connected neurons. These edges can be undirected, indicating a stochastic association, whence the FC is termed as Associative FC (AFC). Alternatively, the edges can be directed and represent a causal relationship between the activities of neurons, whence the FC is termed as Causal FC (CFC). Finding the CFC is expected to facilitate the inference of the governing neural interaction pathways essential for brain functioning and behavior \citep{finn2015functional,biswas2022statistical1}, and serves as a promising biomarker for neuro-psychiatric diseases \citep{ashraf2015cortical,nakamura2017early, brier2014functional, sheline2013resting}. In terms of the framework described in Section \ref{sec:tsscm}, the CFC can be represented by the Rolled Graph of causal neural interaction in the neural time series, which is inferred by CITS algorithm (See Definition \ref{def:rolledgraph}). We now demonstrate application of the methods in this paper to obtain the CFC between neurons from neural time series in a real neurobiological dataset consisting of electrophysiological recordings in the Visual Coding Neuropixels dataset of the Allen Brain Observatory.

\begin{figure}[t!]
    \centering
    \includegraphics[width = 0.9\textwidth]{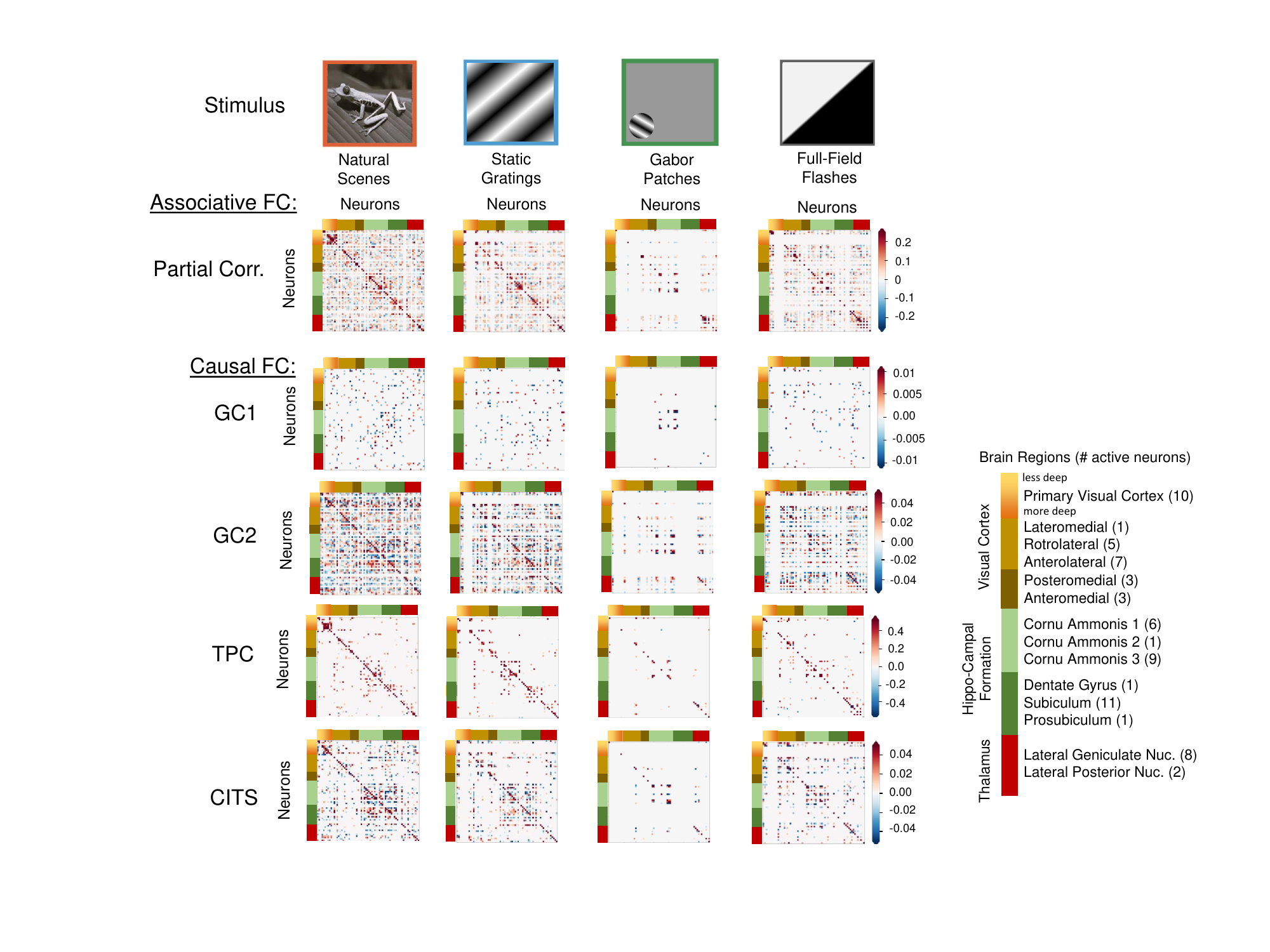}
    \caption{Four different methods for inferring functional connectivity (FC) were compared in a mice brain benchmark data from the Allen Institute's Neuropixels dataset. These methods included Associative FC using Partial Correlation, and Causal FC using GC1, GC2, TPC and CITS. The FC is estimated and represented by an adjacency matrix with edge weights, which is symmetric for Associative FC and asymmetric for Causal FC. In the adjacency matrices, a non-zero entry in $(i, j)$ indicated the connection of neuron $i$ to $j$.}
    \label{fig:resneuropixels}
\end{figure}

In this study, we evaluate CITS alongside Time-Aware PC (TPC), Granger Causality (GC1 and GC2) and Partial Correlation which are well-known methods to obtain CFC and Associative Functional Connectivity (AFC) from electrophysiological recordings. We use the Visual Coding Neuropixels dataset from the Allen Brain Observatory ~\citep{de2020large,allenbrainobs}, which includes sorted spike trains and local field potentials recorded simultaneously across six cortical visual areas, hippocampus, thalamus, and other adjacent structures of the mice brain, while being presented with various stimuli such as natural scenes/images, static gratings, Gabors, and full-field flashes. These stimuli are repeatedly presented to the mice and the data is recorded using the Neuropixels technology, which inserts multiple electrodes into the brain allowing real-time recording from hundreds of neurons across the brain  ~\citep{neuropixels}. For more details of the dataset, see Appendix \ref{datadescriptionfull}.

In Figure \ref{fig:resneuropixels}, we show the adjacency matrices for the functional connectivity (FC) obtained by different methods, each for one trial in a specific stimuli category. The Associative FC (AFC) pattern among neurons in each stimuli category is distinct. The Causal FC (CFC) is expected to be a directed subset of the AFC and consistent with its overall AFC pattern \citep{dadgostar2016functional, wang2016efficient}. Such is observed in the CFC obtained by TPC and CITS, which yield asymmetric adjacency matrices, yet match the overall pattern in the AFC in a sparse and dense manner, respectively. In contrast, the CFC obtained by GC1 is sparse but does not match the patterns in the AFC so well, and the one obtained by GC2 is too dense, thereby prone to noise. In the CFC obtained by both TPC and CITS (refer to Fig \ref{fig:resneuropixels}), natural scenes evoke greater connectivity within active neurons in the Primary Visual Cortex, static gratings evoke greater connectivity in the Posteromedial and Anteromedial Visual Cortex, and full-field flashes evoke more connectivity in the Anterolateral Visual Cortex and Thalamus compared to other stimuli. All the four stimuli exhibit distinct connectivity patterns in the Cornu Ammonis regions of the Hippocampal Formation. In addition, natural scenes and static gratings induce more prominent connectivity within the Subiculum compared to other stimuli.

\section{Discussion}
The paper focuses on stationary Markovian processes that follow a structural causal model and do not have concurrent effects. For such processes, we propose a novel algorithm for detecting the direct causes of the variables at a specific time, coined Time-series Structural Causal Inference (CITS), based on the conditional independence information within a $2\tau$ window before that time. We prove the correctness of the CITS algorithm provided a conditional independence oracle. We show its consistency in estimating the correct DAG with a consistent statistical test for conditional dependence. Such tests include Pearson's partial correlations for the Gaussian regime and the Hilbert-Schmidt Criterion in the non-Gaussian regime. The comparative performance of the CITS algorithm is demonstrated through numerical experiments. We apply the methods to obtain the Causal Functional Connectivity (CFC) between neurons from neural time series, recorded from mice brain presented with different stimuli. The results provide insights into the characteristics of the CFC across a variety of stimuli scenarios.

A valuable and novel aspect of CITS is that it finds the direct causal influences on the variables at a specific time by looking for conditional independence information within a $2\tau$ interval preceding a specific time. In essence, the $2\tau$ window suffices as the Markovian property of order $\tau$ implies that influences will originate from a $\tau$ window preceding the specific time, and influences to the $\tau$ window will originate from a $2\tau$ window preceding the time. Furthermore, in numerical experiments, CITS performed better overall compared to other approaches. In particular, in experiments with a common cause for multiple variables, while all other methods suffer in performance, CITS maintains its close-to-perfect performance due to leveraging the conditional independence information in the $2\tau$ window.

Under the model assumptions of first order Markovian SCM and no concurrent effects, even with the presence of unobserved confounders, the DAG and Rolled Graph corresponding to the SCM can be obtained by CITS devoid of any spurious edges introduced by the unobserved variables. This is because spurious edges arise if and only if an unobserved variable is a common parent to observed variables, and under these assumptions, this can happen if and only if the children are at the same time, e.g. $X_{u,t}$, $X_{v,t}$, which is ruled out by the no concurrent effects assumption. Since the first order Markov property is a reasonable condition for many applications \citep{liu2010application}, it makes this proposition appealing in such scenarios.

CITS provides a powerful tool for causal inference in the time series setting based on stationary Markovian Structural Causal Models. The Markovian assumption can be tested by statistical tests such as in \citep{chen2012testing}. In the non-Markovian scenario, the $2\tau$ technique used by CITS loses its relevance, yet a Markovian assumption is widely used for time series and is often reasonable to make in practice \citep{freedman2012markov}. The stationary assumption can also be tested by statistical tests \citep{witt1998testing, priestley1969test}.  For a non-stationary scenario, the time series can often be considered to be locally stationary \citep{vogt2012nonparametric, last2008detecting}. Then, CITS can be implemented on smaller consecutive windows over time, which outputs a causal graph for each window, thereby yielding a varying causal graph over time.

\section{Methods}
\subsection{Causal inference in time series - Review}\label{sec: cinf_ts_review}
Among methodologies for causal inference in the time series scenario, one of the foremost is Granger Causality \citep{granger2001essays,shojaie2022granger}. Granger Causality became popular as a parametric model-based approach that uses a vector autoregressive model for the time series data whose non-zero coefficients indicate causal effect between variables. In recent years, there have been non-linear extensions to Granger Causality \citep{tank2021neural}. Transfer Entropy is a non-parametric approach equivalent to Granger Causality for Gaussian processes \citep{barnett2009granger}.

A different framework for causal inference is the well-known Directed Graphical Modeling framework, which models causal relationships between variables by the Directed Markov Property with respect to a Directed Acyclic Graph (DAG) \citep{spirtes2000causation,pearl2009causality}. It is a popular framework for independent and identically distributed (i.i.d.) datasets. Inference in this framework is either constraint-based or score-based. Constraint-based methods are based on conditional independence (CI) tests such as the SGS algorithm \citep{verma2022equivalence}, its faster incarnation of PC algorithm \citep{kalisch2007estimating}, both assuming no latent confounders, and the FCI algorithm in the presence of latent confounders \citep{spirtes1999algorithm}. On the other hand, score-based methods perform a search on the space of all DAGs to maximize a goodness-of-fit score, for example, the Greedy Equivalence Search (GES) and Greedy Interventional Equivalence Search (GIES) \citep{hauser2012characterization}. However, these approaches are based on i.i.d. observations, and their extension to time series settings is not trivial. In fact, naive application of PC algorithm to time series data is seen to suffer in performance due to not incorporating across-time causal relations \citep{biswas2022statistical1}.

Recently the PC algorithm has been extended to form the Time-Aware PC algorithm that is applicable to time series datasets \citep{biswas2022statistical2}. Entner and Hoyer suggest to use the FCI algorithm for inferring the causal relationships from time series in the presence of unobserved variables \citep{entner2010causal}. The PCMCI algorithm \citep{runge2019detecting}, another variation of the PC algorithm,  consists of the following two steps. In the first step, it uses the PC algorithm to detect parents of a variable at a time $t$, and in the second step, it applies a momentary conditional independence (MCI) test. Chu et al. \citep{chu2008search} propose a causal inference algorithm based on conditional independence, designed for additive, nonlinear and stationary time series, using two kinds of conditional dependence tests based on additive regression model. Using the fact that one can incrementally construct the mutual information between a cause and an effect, based on the mutual information values between the effect and the previously found causes, Jangyodsuk et al. \citep{jangyodsuk2014causal} proposes to obtain a causal graph with each time series as a node and the edge weight for each edge is the difference in time steps between the occurrence of cause and the occurrence of the effect. Amortized Causal Inference is another conditional independence-based algorithm, that infers causal structure by performing amortized variational inference over an arbitrary data-generating distribution \citep{lorch2022amortized}.
\subsection{Mathematical Guarantee of the CITS algorithm}
Let $\mu(X_{u,s},X_{v,t}\vert\bm{X}_S)$ denote a measure of conditional dependence of $X_{u,s}$ and $X_{v,t}$ given $\bm{X}_S$, i.e. it takes value $0$ if and only if we have conditional independence, and let $\hat{\mu}_n(X_{u,s},X_{v,t}\vert\bm{X}_S)$ be its consistent estimator. In Theorem \ref{thm:main}, we will use $\hat{\mu}_n(X_{u,s},X_{v,t}\vert\bm{X}_S)$ to construct a conditional dependence test guaranteeing consistency of the CITS estimate.

\begin{theorem}\label{thm:main}
    Let $\{\bm{X}_t\}_{t\in \mathbb{Z}}$ be a strictly stationary Markovian process of order $\tau$, and assume that it satisfies a structural causal model with time invariant DAG $G$ as in \eqref{eqdef: process}, which exhibits no concurrent effects and whose distribution is faithful to $G$. Denote by $\hat{G}_n$ and $\hat{G}_{n,R}$, respectively, the estimates of $G$ and its Rolled Graph $G_R$ obtained by sample CITS (\ref{alg:CITS_sample}) based on $n$ time samples and using the conditional dependence test $\left\vert\hat{\mu}_n(X_{u,s},X_{v,t}\vert\bm{X}_S)\right\vert>\gamma$ for some fixed $\gamma>0$, where  $\hat{\mu}_n(X_{u,s},X_{v,t}\vert\bm{X}_S)$ is a consistent estimator of a conditional dependence measure $\mu(X_{u,s},X_{v,t}\vert\bm{X}_S)$. Then we have,
    \[
    P(\hat{G}_{n}=G)\rightarrow 1 \quad\text{ as }\quad n\rightarrow \infty.
    \]
    \[
    P(\hat{G}_{n,R}=G_R)\rightarrow 1 \quad\text{ as }\quad n\rightarrow \infty.
    \]
\end{theorem}

See Appendix \ref{proofTh4.1} for the proof of Theorem \ref{thm:main}. In this section, we are going to consider two consistent estimators of the conditional dependence measures, namely the (log-transformed version of the) partial correlation $Z_n$ in the Gaussian setting, and the Hilbert-Schmidt estimator $H_n$ (see Sections 2.3 and 2.4  in \citep{biswas2022consistent}) in the non-Gaussian setting. For a detailed account of these conditional dependence estimators, see Appendix \ref{appsec:condindtests}. 

In order to establish consistency of the estimators $Z_n$ and $H_n$, under the non-i.i.d. setting, one needs to assume some mixing conditions on the underlying time series. Two such standard mixing conditions are $\rho$-mixing and $\alpha$-mixing \citep{kolmogorov1960strong,bradley2005basic} (see Appendix \ref{defmixing} for the definitions of these two notions of mixing).  We make either of the following two assumptions related to mixing of the time series to ensure consistency of the estimators.

\subsection*{Assumption 1}\label{rhoassumptions}
\begin{enumerate}
    \item $\{X_{u,t},X_{v,t}: t=1,2,\ldots\}$ is $\rho$-mixing for all $u,v$, with maximal correlation coefficients $\xi_{uv}(k), k\ge 1$.
     \item $\e X_{u,t}^4 <\infty$ for all $u,t$ and $\sum_{k=1}^{\infty} \xi_{uv} (k) < \infty$ for $u,v\in 1,\dots,p$.
    \item There exists a sequence of positive integers $s_n\rightarrow \infty$ and $s_n = o(n^{1/2})$ such that $n^{1/2}\xi_{uv}(s_n)\rightarrow 0$ as $n\rightarrow\infty$ for $u,v\in 1,\ldots,p$.
\end{enumerate}

\subsection*{Assumption 1*}\label{alphaassumptions}
\begin{enumerate}
    \item[$1^*$.] $\{X_{u,t},X_{v,t}: t=1,2,\ldots\}$ is strongly mixing with coefficients $\{\alpha_{uv}(k)\}_{k\geq 1}$ for all $u,v=1,\ldots,p$.
    \item[$2^*$.] $\e |X_{v,t}|^{2\delta} < \infty$ for some $\delta > 2$ and all $v,t$, and the strongly mixing coefficients satisfy: $\sum_{k=1}^{\infty} \alpha_{uv} (k)^{1-2/\delta} < \infty$ for $u,v= 1,\dots,p$.
    \item[$3^*$.] There exists a sequence of positive integers $s_n\rightarrow \infty$ and $s_n = o(n^{1/2})$ such that $n^{1/2}\alpha_{uv}(s_n)\rightarrow 0$ as $n\rightarrow\infty$ for $u,v= 1,\ldots,p$.
\end{enumerate}
Assumptions 1.1-1.3 and $1^*.1^*-1^*.3^*$ are adapted from Conditions 1-2 in \citep{masry2011estimation}. See Remark 4.1 in \citep{biswas2022consistent} for a comparative discussion between these two sets of assumptions.

\begin{lemma}\label{lemma:consistency}
Under either Assumption \hyperref[rhoassumptions]{1} or Assumption \hyperref[alphaassumptions]{$1^*$}, 
\begin{enumerate}
    \item $Z_n(X_{u,s},X_{v,2\tau+1}\vert X_{S})$ converges to $z(X_{u,s},X_{v,2\tau+1}\vert X_{S})$ in probability.
    
    \item If the regularization constant $\epsilon_n$ satisfies $n^{-1/3} \ll \epsilon_n \ll 1$, then, $H_n(X_{u,s},X_{v,2\tau+1}\vert X_{S})$ converges to $H(X_{u,s},X_{v,2\tau+1}\vert X_{S})$ in probability,
\end{enumerate}
where $u,v=1,\ldots,p$, $s=\tau+1,\ldots,2\tau$ and $S\subset\{(d,r):d=1,\ldots,p;r=1,\ldots,2\tau+1\}\setminus\{(u,s),(v,2\tau+1)\}$.
\end{lemma}
\begin{proof}
The proof directly follows from Lemmas 3.3 and 3.4 in \citep{biswas2022consistent}. 
\end{proof}
\begin{corollary}
Under Assumption \hyperref[rhoassumptions]{1} or Assumption \hyperref[alphaassumptions]{$1^*$}, using sample CITS based on either of these two conditional dependence tests: 1) partial correlation for Gaussian regime and 2) Hilbert-Schimdt conditional dependence criterion for non-Gaussian regime, leads to asymptotically accurate estimation of the DAG $G$ and Rolled DAG $G_R$ (see Theorem \ref{thm:main}). This is because, the sample z-transformed partial correlation for the Gaussian regime and the Hilbert-Schimdt conditional dependence criterion for the non-Gaussian regime, form consistent estimators to their corresponding conditional dependence measures (see Lemma \ref{lemma:consistency}).
\end{corollary}

\subsection{Edge weights}\label{sec:edgeweights}
In this section, we first assign an edge weight, denoted $w_{u,s}^{v,t}$, for the edge $X_{u,s}\rightarrow X_{v,t}$ in the Unrolled DAG estimate $\hat{G}$ obtained by CITS.
\begin{equation}\label{eq:causaleffectgaussian}
    w_{u,s}^{v,t}=\left\{\begin{array}{lr}
    0, & \text{ if } X_{u,s} \not\in pa_{\hat{G}}(X_{v,t}),\\
    \text{coefficient of }X_{u,s} \text{ in }X_{v,t} \sim pa_{\hat{G}}(X_{v,t}) & \text{ if } X_{u,s} \in pa_{\hat{G}}(X_{v,t})\end{array}\right.
\end{equation}
where $X_{v, t} \sim pa_{\hat{G}}(X_{v,t})$ is a shorthand for linear regression of $X_{v,t}$ on $pa_{\hat{G}}(X_{v,t}) = \{X_{a,b} : X_{a,b} \rightarrow X_{v,t} \text{ is an edge in } \hat{G}\}$. That is, the edge weights are obtained based on linear regressions corresponding to the estimated DAG of the SCM.

Next, we assign an edge \textit{weight} for the edge $u\rightarrow v$ in the Rolled graph estimate $\hat{G}_R$, following the way $\hat{G}_{R}$ is defined from $\hat{G}$ in the CITS algorithm: if nodes $u,v$ are connected as $u\rightarrow v$ in $\hat{G}_R$, then, its edge weight, denoted by $w_u^v$, is defined as the average of the edge weights: $w_{u,s}^{v,2\tau+1}$ for $X_{u,s}\in \text{pa}(X_{v,2\tau+1})$ in $\hat{G}$, $s=\tau+1,\ldots,2\tau$.

\section{Code Availability}
The software package for the methods in this paper and example code are available in \url{https://github.com/biswasr/cits}.

\bibliography{main}

\appendix
\section{Proofs of the Main Results}
\subsection{Proof of Lemma \ref{lemma:concept}}\label{proofLem3.1}
To see that (1) implies (2), let (2) be false, i.e. $X_{v,t}$ and $X_{u,s}$ are adjacent in $G$. Then by the time order it follows that $X_{u,s}\rightarrow X_{v,t}$, which implies (1) is false. 
    
(2) implies (1) holds trivially. 
    
Next, suppose that (2) holds, i.e. $X_{v,t}$ and $X_{u,s}$ are non-adjacent in $G$. Then $X_{v,t}$ and $X_{u,s}$ are d-separated by the set of their parents in $G$. Next, note that the parents of $X_{v,t}$ and $X_{u,s}$ are between times $\{t-\tau,\ldots,t-1\}$ and $\{s-\tau,\ldots,s-1\}$ respectively and hence, the parents of both $X_{v,t}$ and $X_{u,s}$ are within times $\{t-2\tau,\ldots,t-1\}$. This implies that $X_{v,t}$ and $X_{u,s}$ are d-separated by $S_0:= \text{pa}(v,t)\cup \text{pa}(u,s) \subseteq \{(d,r):d=1,\ldots,p; r= t-2\tau,\ldots,t-1\}$. Since the SCM implies directed Markov property \citep{lauritzen1996graphical,bollen1989structural}, so it follows that $X_{v,t}\ind X_{u,s} \vert \bm{X}_{S_0}$. We thus showed that (2) implies (3).

Finally, let (3) hold. Then under faithfulness, it follows that $X_{v,t}$ and $X_{u,s}$ are d-separated in $G$ by $S$ which implies (2). 

\subsection{Proof of Theorem \ref{thm:main}}\label{proofTh4.1}
For a pair of nodes $u,v=1,\ldots,p$ and times $s=\tau+1,\ldots,2\tau+1$ and a conditioning set $S\subseteq \mathbb{S}:=\{(d,r):d=1,\ldots,p;r=1,\ldots,2\tau+1\}\setminus\{(u,s),(v,2\tau+1)\}$, let $E_{u,v,s\vert S}$ denote an error event that occurred when testing conditional dependence of $X_{v,2\tau+1}\ind X_{u,s}\vert \bm{X}_S$, i.e.,

\[
E_{u,v,s\vert S} = E_{u,v,s\vert S}^{I} \cup  E_{u,v,s\vert S}^{II},
\]

where
\[
E_{u,v,s\vert S}^{I} := \{|\hat{\mu}_n(X_{u,s},X_{v,2\tau+1}\vert\bm{X}_S)| > \gamma \text{ and } \mu(X_{u,s},X_{v,2\tau+1}\vert\bm{X}_S) = 0\}\]
\[E_{u,v,s\vert S}^{II} := \{ |\hat{\mu}_n(X_{u,s},X_{v,2\tau+1}\vert\bm{X}_S)| \le \gamma \text{ and } \mu(X_{u,s},X_{v,2\tau+1}\vert\bm{X}_S) \neq 0\}
\] denote the events of Type I error and Type II error, respectively. Thus,

\begin{align*}
P(\text{an error occurs in the sample CITS algorithm})&\leq P\left(\bigcup_{u,v,s,S\subseteq \mathbb{S}}E_{u,v,s\vert S}\right)\\
& \leq O(1)\sup_{u,v,s,S\subseteq\mathbb{S}}P(E_{u,v,s\vert S}) \numberthis\label{eqn:skelerror}   
\end{align*}
using that the cardinality of the set $\vert \{u,v,s,S\subseteq \mathbb{S}\} \vert = p^2\tau 2^{2p\tau-2}$. Then, for any $\gamma > 0$, we have:

\begin{equation}\label{t1err}
    \sup_{u,v,s,S\subseteq\mathbb{S}} P(E_{u,v,s\vert S}^{I}) \le \sup_{u,v,s,S\subseteq \mathbb{S}} P(\vert \hat{\mu}_n(X_{u,s},X_{v,2\tau+1}\vert\bm{X}_S) - \mu(X_{u,s},X_{v,2\tau+1}\vert\bm{X}_S)\vert > \gamma) = o(1)
\end{equation}
by the consistency of $\hat{\mu}_n$.

Next, we bound the type II error probability. Towards this,  let $c=\inf \{\vert\mu(X_{u,s},X_{v,t}\vert\bm{X}_S)\vert : \mu(X_{u,s},X_{v,t}\vert\bm{X}_S)\neq 0, u,v=1,\ldots,p;s=\tau+1,\ldots,2\tau+1;S\subseteq \{(d,r):d=1,\ldots,p;r=1,\ldots,2\tau+1\}\setminus\{(u,s),(v,2\tau+1)\}\}>0$, and choose $\gamma = c/2$. Then,
	
\begin{align}
    \sup_{u,v,s,S\subseteq\mathbb{S}} P(E_{u,v,s\vert S}^{II}) &= \sup_{u,v,s,S\subset \mathbb{S}} P(\vert \hat{\mu}_n(X_{u,s},X_{v,2\tau+1}|S)\vert \leq \gamma ~, ~\mu(X_{u,s},X_{v,2\tau+1}|S)\neq 0)\nonumber\\
&\leq \sup_{u,v,s,S\subset \mathbb{S}} P(\vert \hat{\mu}_n(X_{u,s},X_{v,2\tau+1}|S) - \mu(X_{u,s},X_{v,2\tau+1}|S) \vert \ge c/2))\nonumber\\ &=o(1)\label{eqn:skeltype2err}.
\end{align}

It follows from \eqref{eqn:skelerror}, \eqref{t1err} and \eqref{eqn:skeltype2err}, that:
\begin{equation}\label{eq: main_last}
P(\text{an error occurs in the sample CITS algorithm})\rightarrow 0
\end{equation}
The event of no error ocurring in the sample CITS algorithm is same as the event that the outcome of sample CITS would be same as the CITS-Oracle which has knowledge of conditional independence information. By Theorem \ref{thm:CITS_oracle}, the CITS-Oracle outputs the true DAG $G$ and its Rolled DAG $G_R$. In summary, the event of no error occurring in the sample CITS algorithm implies is equivalent to stating that $\hat{G}=G$ and $\hat{G}_R=G_R$. The proof of Theorem \ref{thm:main} is now complete, in view of \eqref{eq: main_last}.

\section{Choice of Conditional Dependence Tests}\label{appsec:condindtests}
In this section, we describe some choices of the conditional dependence tests used in the CITS algorithm. According to Theorem \ref{thm:main}, we can use a conditional dependence test of the following form \eqref{eq:cond_test} in the CITS algorithm, 
\begin{equation}\label{eq:cond_test}
\vert\hat{\mu}_n(X_{u,s},X_{v,t}\vert\bm{X}_S)\vert>\gamma
\end{equation} to guarantee its consistency, as long as it satisfies the condition: $\hat{\mu}_n(X_{u,s},X_{v,t}\vert\bm{X}_S)$ is a consistent estimator of $\mu(X_{u,s},X_{v,t}\vert\bm{X}_S)$, the latter being a measure of conditional dependence of $X_{u,s}$ and $X_{v,t}$ given $\bm{X}_S$. In the following sections, we provide examples of such candidates for $\hat{\mu}_n(X_{u,s},X_{v,t}\vert\bm{X}_S)$ and resulting conditional dependence tests, in both the Gaussian as well as the non-Gaussian regime.
\subsection{The Gaussian Regime: Pearson's Partial Correlations}\label{sec:samplepc} 
It is popular to use partial correlations to test conditional dependence for causal inference, such as in the PC algorithm \citep{kalisch2007estimating}. The partial correlation-based conditional dependence test is applicable in the Gaussian setting. Assume that $\mathbf{Y}=(Y_1,\ldots,Y_p)$ is a $p$-dimensional Gaussian random vector, for some fixed integer $p$. For $i \neq j \in \{1,\ldots ,p\},\ \bm{k} \subseteq \{1,\ldots ,p\}
  \setminus 
\{i,j\}$, denote by $\rho(Y_i,Y_j|Y_{\bm{k}})$ the 
partial correlation between $Y_i$ and $Y_j$ given
$\{Y_r:\ r \in \bm{k}\}$. The partial correlation serves as a measure of conditional dependence in the Gaussian regime in view of the following standard property of the multivariate Gaussian distribution (see Prop. 5.2 in \citep{lauritzen1996graphical}),

\[\rho(Y_i,Y_j|Y_{\bm{k}})=0\text{ if and only if }Y_i\ind Y_j ~\vert~ \{Y_r~:\ r \in \bm{k}\}.\] 

\medskip
Denote $k=\vert \bm k\vert$ and let without loss of generality $\{Y_r;\ r \in \bm{k}\}$ be the last $k$ entries in $\bm{Y}$. Let $\Sigma:= \text{cov}(\bm{Y})$ with $\Sigma = \left(\begin{array}{cc}
    \Sigma_{11} & \Sigma_{12} \\
    \Sigma_{21} & \Sigma_{22}
\end{array}\right)$ where $\Sigma_{11}$ is of dimension $(p-k)\times (p-k)$, $\Sigma_{22}$ is of dimension $k \times k$, and $\Sigma_{11.2}=\Sigma_{11}-\Sigma_{12}\Sigma_{22}^{-1}\Sigma_{21}$. Let $\bm e_1,\ldots,\bm e_p$ be the canonical basis vectors of $\mathbb{R}^p$. It follows from \citep{muirhead2009aspects} that,
\begin{align*}
\rho(Y_i,Y_j|Y_{\bm{k}}) &= \frac{\bm{e}_i^{\top}\Sigma_{11.2}\bm{e}_j}{\sqrt{(\bm{e}_i^{\top}\Sigma_{11.2}\bm{e}_i(\bm{e}_j^{\top}\Sigma_{11.2}\bm{e}_j))}}
\end{align*}
One can calculate the sample partial correlation
$\hat{\rho}(Y_i,Y_j|Y_{\bm{k}})$ via regression or by using the following identity, with  $\hat{\Sigma}$ and $\hat{\Sigma}_{11.2}$ being the sample versions of $\Sigma$ and $\Sigma_{11.2}$.
\begin{align*}
\hat{\rho}(Y_i,Y_j|Y_{\bm{k}}) &= \frac{\bm{e}_i^{\top}\hat{\Sigma}_{11.2}\bm{e}_j}{\sqrt{(\bm{e}_i^{\top}\hat{\Sigma}_{11.2}\bm{e}_i)(\bm{e}_j^{\top}\hat{\Sigma_{11.2}}\bm{e}_j)}}
\end{align*}
For testing whether a partial correlation is zero or not, we first apply Fisher's z-transform
\begin{eqnarray}\label{ztrans}
Z_n(Y_i,Y_j|Y_{\bm{k}}):= g(\hat{\rho}(Y_i,Y_j|Y_{\bm{k}})) :=\frac{1}{2} \log \left (\frac{1 +
    \hat{\rho}(Y_i,Y_j|Y_{\bm{k}})}{1 - \hat{\rho}(Y_i,Y_j|Y_{\bm{k}})} \right).
\end{eqnarray}

Also, let $z(Y_i,Y_j|Y_{\bm{k}}) = g(\rho(Y_i,Y_j|Y_{\bm{k}}))$. Note that $z(Y_i,Y_j|Y_{\bm{k}})=0 \Leftrightarrow \rho(Y_i,Y_j|Y_{\bm{k}})=0$, and hence, $z(Y_i,Y_j|Y_{\bm{k}}) = 0 \Leftrightarrow Y_i\ind Y_j ~\vert~ \{Y_r~:\ r \in \bm{k}\}$. That is, $z(Y_i,Y_j|Y_{\bm{k}})$ is also a measure of conditional dependence. Furthermore, $Z_n(Y_i,Y_j|Y_{\bm{k}})$ is a consistent estimator of the conditional dependence measure $z(Y_i,Y_j|Y_{\bm{k}})$ (See Lemma \ref{lemma:consistency}). Hence, it can be used to construct a conditional dependence test of the form $Z_n(Y_i,Y_j|Y_{\bm{k}}) > \gamma$ for some fixed $\gamma>0$, for the CITS in the Gaussian regime.

\subsection{The Non-Gaussian Regime: The Hilbert Schmidt Criterion.}\label{sec:samplepch}

In the general non-Gaussian scenario, zero partial correlations do not necessarily imply conditional independence. In such cases, the Hilbert Schmidt criterion can be used as a convenient test for conditional dependence, which is described in more details below.

Given $\mathbb{R}$-valued random variables $X,Y$ and the random vector $\bm Z$, we propose to use the following statistic for testing the conditional dependence of $X,Y\vert \bm{Z}$ (see \citep{fukumizu2007kernel}):

\[
H_{n}(X,Y\vert\bm{Z}) = Tr[R_{\overset{..}{Y}}R_{\overset{..}{X}}-2R_{\overset{..}{Y}}R_{\overset{..}{X}}R_{\bm Z} + R_{\overset{..}{Y}}R_{\bm Z}R_{\overset{..}{X}}R_{\bm Z}]
\]
where $R_X := G_X(G_X+n\epsilon_n I_n)^{-1}, R_Y := G_Y(G_Y+n\epsilon_n I_n)^{-1}, R_{\bm Z} := G_{\bm Z}(G_{\bm Z}+n\epsilon_n I_n)^{-1}$, $G_X,G_Y,G_{\bm Z}$ being the centered Gram matrices with respect to a positive definite and integrable kernel $k$, i.e. $G_{X,ij}=<k(\cdot,X_i)-\hat{m}_X^{(n)},k(\cdot,X_j)-\hat{m}_X^{(n)}>$ with $\hat{m}_X^{(n)}=\frac{1}{n}\sum_{i=1}^n k(\cdot,X_i)$, and $\overset{..}{X} := (X,{\bm Z}), \overset{..}{Y} := (Y,{\bm Z})$. Under some regularity assumptions mentioned below, it follows from the proof of Theorem 5 in \citep{fukumizu2007kernel} that $\hat{H}_n(X,Y\vert \bm{Z})$ is a consistent estimator of $H(X,Y\vert \bm Z):= \|V_{\overset{..}{X}\overset{..}{Y}|{\bm Z}}\|^2$, where $$V_{\overset{..}{X}\overset{..}{Y}|{\bm Z}} := \Sigma_{\overset{..}{X}\overset{..}{X}}^{-1/2}(\Sigma_{\overset{..}{X}\overset{..}{Y}} - \Sigma_{\overset{..}{X}\bm Z} \Sigma_{\bm Z \bm Z}^{-1}\Sigma_{\bm Z \overset{..}{Y}}) \Sigma_{\overset{..}{Y} \overset{..}{Y}}^{-1/2}$$ and $\Sigma_{UV}$ denotes the covariance matrix of $U$ and $V$. It follows from \citep{fukumizu2007kernel} that $X\ind Y |\bm Z \Leftrightarrow H(X,Y|\bm Z)=0$, i.e. $H(X,Y|\bm Z)$ is a measure of conditional dependence of $X$ and $Y$ given $Z$. Also, it follows from Lemma \ref{lemma:consistency} that $H_n(X,Y| \bm Z)$ is a consistent estimator to $H(X,Y|\bm Z)$, and hence, $H_n(X,Y| \bm Z)$ can be used to construct a conditional dependence test in the CITS of the form: $\vert H_n(X,Y| \bm Z)\vert > \gamma$ for a fixed $\gamma>0$. This method has the advantage that unlike the Pearson partial correlation, it does not require Gaussianity of the data to decide conditional independence, and consequently can be used in CITS if the underlying time series is non-Gaussian.

\section{Two Notions of Mixing}\label{defmixing}
For fixed $u,v \in 1,\ldots, p$, let $\mathcal{F}_{a}^{b}$ be the $\sigma$-field of events generated by the random variables $\{X_{u,t},X_{v,t}:a\leq t \leq b\}$, and $L_2(\mathcal{F}_a^b)$ be the collection of all second-order random variables which are $\mathcal{F}_a^b$-measurable.

\begin{definition}[$\rho$-mixing]
In this section, we describe two common notions of mixing, that we are going to assume on our underlying time series, in order to guarantee consistency of the conditional dependence estimators. The stationary process $\{X_{u,t},X_{v,t}:t= 1,2,\ldots\}$ is called $\rho$-mixing if the maximal correlation coefficient

\begin{equation}\label{eq:rhomixing}
\xi_{uv}(k):=\sup_{\ell\geq 1}\sup_{\substack{U\in L_2(\mathcal{F}_{1}^\ell)\\ V\in L_2(\mathcal{F}_{\ell+k}^{\infty})}} \frac{\vert \text{cov}(U,V) \vert}{\text{var}^{1/2}(U)\text{var}^{1/2}(V)} \rightarrow 0 \text{ as } k\rightarrow \infty.
\end{equation}
\end{definition}

\begin{definition}[$\alpha$-mixing]
The stationary process $\{X_{u,t},X_{v,t}:t= 1,2,\ldots\}$ is called $\alpha$-mixing if:

\[
\alpha_{uv}(k):=\sup_{\ell\geq 1}\sup_{\substack{A\in L_2(\mathcal{F}_{1}^\ell)\\ B\in L_2(\mathcal{F}_{\ell+k}^{\infty})}} |P(A\cap B) - P(A) P(B)| \rightarrow 0 \text{ as } k\rightarrow \infty.
\]
\end{definition}
\section{Simulation Study Details}\label{appen:simulstudy}
We outline the different simulation settings used in the numerical experiments.

\noindent\textbf{Linear Additive with Gaussian noise:}
\begin{itemize}
    \item Linear Gaussian Model 1: 
    $$(X_{1,t},X_{2,t},X_{3,t},X_{4,t}) = ( 1+\epsilon_{1,t},~-1+\epsilon_{2,t},~2X_{1,t-1}-X_{2,t-1}+\epsilon_{3,t}, ~2X_{3,t-1}+\epsilon_{4,t})$$
     where $\epsilon_{i,t}, i=1,\ldots,4, t=1,2,\ldots,1000, \sim$ i.i.d. $N(0,\eta)$ with mean $0$ and standard deviation $\eta$. In this and all subsequent examples, the parameter $\eta$ is assumed to vary between $0$ to $3.5$ in increments of $0.5$ in our numerical experiments. The true Rolled graph $G_R$ for this model has the edges $1\rightarrow 3, 2\rightarrow 3, 3\rightarrow 4$.
    \item  Linear Gaussian Model 2: 
    $$(X_{1,t},X_{2,t},X_{3,t},X_{4,t}) = ( 1+\epsilon_{1,t},~-1+2X_{1,t-1}+\epsilon_{2,t},~2X_{1,t-1}+\epsilon_{3,t}, ~X_{2,t-1} + X_{3,t-1} +\epsilon_{4,t})$$
   where $\epsilon_{i,t}, i=1,\ldots,4, t=1,2,\ldots,1000, \sim$ i.i.d. $N(0,\eta)$ with mean $0$ and standard deviation $\eta$. The true Rolled graph $G_R$ for this model has the edges $1\rightarrow 2, 1\rightarrow 3, 2\rightarrow 4, 3\rightarrow 4$.
   \end{itemize}
    
    \noindent\textbf{Non-linear Additive with Non-Gaussian noise:}
    \begin{itemize}
    \item Non-linear Non-Gaussian Model 1:     
    $$(X_{1,t},X_{2,t},X_{3,t},X_{4,t}) = (\epsilon_{1,t},~\epsilon_{2,t},~4\sin(X_{1,t-1})-3\sin(X_{2,t-1})+\epsilon_{3,t},~ 3 X_{3,t-1} +\epsilon_{4,t})$$
   where $\epsilon_{i,t}, i=1,\ldots,4, t=1,2,\ldots,1000, \sim$ i.i.d. and uniformly distributed on the interval $(0,\eta)$. The true Rolled graph $G_R$ for this model has the edges $1\rightarrow 3, 2\rightarrow 3, 3\rightarrow 4$. 
    
    \item Non-linear Non-Gaussian Model 2:  
    $$(X_{1,t},X_{2,t},X_{3,t},X_{4,t}) = ( \epsilon_{1,t},~4X_{1,t-1}+\epsilon_{2,t},~3\sin(X_{1,t-1})+\epsilon_{3,t}, ~8\log(\vert X_{2,t-1}\vert) + 9\log(\vert X_{3,t-1}\vert) +\epsilon_{4,t})$$
   where $\epsilon_{i,t}, i=1,\ldots,4, t=1,2,\ldots,1000, \sim$ i.i.d. Uniform distribution on the interval $(0,\eta)$. The true Rolled graph $G_R$ for this model has the edges $1\rightarrow 2, 1\rightarrow 3, 2\rightarrow 4, 3\rightarrow 4$.
\end{itemize}

   \noindent \textbf{Non-linear Non-additive Continuous Time Recurrent Neural Network with Gaussian noise:}
   \begin{itemize}
   \item Continuous Time Recurrent Neural Network (CTRNN) Model:  \begin{equation}\label{ctrnn}
    \tau_j \frac{dX_{j,t}}{dt}=-X_{j,t} + \sum_{i=1}^m w_{ij} \sigma (X_{i,t}) + \epsilon_{j,t}, j=1,\ldots, m,
    \end{equation}
    We consider a motif consisting of $m=4$ and $w_{13}=w_{23}=w_{34}=10$ and $w_{ij}=0$ otherwise. We also note that in Eq. \eqref{ctrnn}, $X_{j,t}$ depends on its own past. Therefore, the true Rolled graph has the edges $1\rightarrow 3,2\rightarrow 3,3\rightarrow 4, 1\rightarrow 1, 2\rightarrow 2, 3\rightarrow 3, 4\rightarrow 4$. The time constant $\tau_i$ is set to 10 units for each $i$. We consider $\epsilon_{i,t}$ to be distributed as an independent Gaussian process with mean $1$ and standard deviation $\eta$. The signals are sampled at a time gap of $g := \exp(1) \approx 2.72$ units for a total duration of $1000$ units. For simulation purposes, one may replace the continuous derivative on the left hand side of Eq. \eqref{ctrnn} is replaced by first order differences at a gap of $g$. 
\end{itemize}
\section{Description of the Visual Coding Neuropixels Dataset}\label{datadescriptionfull}

In this section, we provide more details on the neuropixels dataset considered in Section \ref{sec:application}. 

\subsection{The Dataset} We restrict our analysis to a 116 days old male mouse (Session ID 791319847), having 555 neurons recorded simultaneously by six Neuropixel probes. The spike trains for this experiment were recorded at a frequency of 1 KHz. The spike trains of this mouse are then studied across four types of stimuli: natural scenes, static gratings, Gabor patches, and full-field flashes. 

The set of stimuli ranges from natural scenes that evoke a mouse's natural habitat, to artificial stimuli such as static gratings, Gabor patches, and full-field flashes. Static gratings consist of sinusoidal patches, Gabor patches consist of sinusoidal patches with decreasing luminosity as the distance from the center increases, and full-field flashes consist of sudden changes in luminosity across the entire visual field. For more details on each of these four types of stimuli, see Section \ref{datadescription}. By using these four stimuli, we aim to investigate how they elicit different patterns of neuronal connectivity. Dynamic stimuli such as natural movies and drifting gratings are excluded from our analysis, since their results need a more nuanced study, which we plan to conduct in the future.

\subsection{Preprocessing} We transformed the recorded spike trains from 1 KHz to a bin size of 10 ms by grouping them based on the start and end times of each stimulus presentation. This allowed us to obtain the Peri-Stimulus Time Histograms (PSTHs) with a bin size of 10 ms. We then used a Gaussian smoothing kernel with a bandwidth of 16ms to smooth the PSTHs for each neuron and each stimulus presentation. The smoothed PSTHs were used as input for inferring the functional connectivity (FC) between neurons for each stimulus presentation. To select the active neurons for each stimulus, we first chose the set of neurons that were active in at least 25\% of the bins in the PSTH, and then collected the unique set of neurons over all stimuli. We found that there were 54, 43, 19, and 36 active neurons for natural scenes, static gratings, Gabor patches, and flashes, respectively, and a total of 68 unique active neurons overall. We separated the entire duration of stimulus presentation to obtain 58 trials of natural scenes, 60 trials of static gratings, 58 trials of Gabor patches, and 3 trials of flashes, where each trial lasted for 7.5 s.

\subsection{Description of the Stimuli}\label{datadescription}
 In this section, we give a detailed description of each of the four stimuli considered in the neuropixels dataset.
\begin{enumerate}
    \item \textit{Natural scenes}, consisting of 118 natural scenes selected from three databases (Berkeley Segmentation Dataset, van Hateren Natural Image Dataset, and McGill Calibrated Colour Image Database) as one of the stimuli. Each scene is displayed for 250ms, after which it is replaced by the next scene in the set. The scenes are repeated 50 times in a random order with blank intervals in between.
    \item \textit{Static gratings} are full-field sinusoidal gratings with 6 orientations, 5 spatial frequencies, and 4 phases, resulting in 120 stimulus conditions. Each grating lasts for 250ms and then replaced with a different orientation, spatial frequency, and phase condition. Each grating is repeated 50 times in random order with blank intervals in between.
    \item \textit{Gabor patches}, for which the patch center is located at one of the points in a 9 × 9 visual field and three orientations are used. Each patch is displayed for 250ms, followed by a blank interval, and then replaced with a different patch. Each patch is repeated 50 times in random order with intermittent blank intervals.
    \item \textit{Full-field flashes}, that last for 250 ms and are followed by a blank interval of 1.75 s before the next flash appears. This sequence is repeated a total of 150 times.
\end{enumerate}

\end{document}